\documentclass[11pt]{article}
\usepackage[margin=1in]{geometry}
\usepackage{amsmath,amsfonts,amssymb,amsthm}
\usepackage{graphicx}
\usepackage{enumerate}
\usepackage{bbm}
\usepackage{verbatim}
\usepackage{hyperref,color}
\usepackage[capitalize,nameinlink]{cleveref}
\usepackage[dvipsnames]{xcolor}
\hypersetup{
	colorlinks=true,
	pdfpagemode=UseNone,
    citecolor=OliveGreen,
    linkcolor=NavyBlue,
    urlcolor=Magenta,
	pdfstartview=FitW
}
\usepackage{appendix}
\crefname{appsec}{Appendix}{Appendices}
\usepackage{tikz}

\theoremstyle{plain}
\newtheorem{theorem}{Theorem}[section]
\newtheorem{proposition}[theorem]{Proposition}
\newtheorem{lemma}[theorem]{Lemma}

\theoremstyle{definition}
\newtheorem{definition}[theorem]{Definition}

\newtheorem*{assumption*}{Assumption}

\theoremstyle{remark}

\crefname{lemma}{Lemma}{Lemmas}
\crefname{theorem}{Theorem}{Theorems}
\crefname{definition}{Definition}{Definitions}
\crefname{fact}{Fact}{Facts}
\crefname{claim}{Claim}{Claims}
\crefname{proposition}{Proposition}{Propositions}

\newcommand{\E}{\mathbb{E}}

\DeclareMathOperator*{\argmax}{arg\,max}
\newcommand{\norm}[1]{\left\lVert #1 \right\rVert}

\newcommand{\ceil}[1]{\left\lceil #1 \right\rceil}

\newcommand{\poly}{\mathrm{poly}}
\newcommand{\dist}{\mathrm{dist}}

\newcommand{\eps}{\varepsilon}

\newcommand{\ball}{\mathsf{B}}

\newcommand{\infmax}{\textsc{Inf-Max}}

\newcommand{\N}{\mathbb{N}}

\newcommand{\R}{\mathbb{R}}

\renewcommand{\Pr}{\mathbb{P}}

\newcommand{\f}{f}

\newcommand{\MM}{\mathcal{M}}

\begin{document}
	
\title{Influence Maximization in Ising Models}
\author{
Zongchen Chen\thanks{Department of Computer Science and Engineering, University at Buffalo, \texttt{zchen83@buffalo.edu}. Research supported by EM's Simons Investigator award (622132).}
\and
Elchanan Mossel\thanks{Department of Mathematics, MIT, \texttt{elmos@mit.edu}. Research supported by the Vannevar Bush Faculty Fellowship ONR-N00014-20-1-2826, the NSF award CCF 1918421, and the Simons Investigator award (622132).}
}
\date{\today}

\maketitle

\begin{abstract}
Given a complex high-dimensional distribution over $\{\pm 1\}^n$, 
what is the best way to increase the expected number of $+1$'s by controlling the values of only a small number of variables?
Such a problem is known as influence maximization and has been widely studied in social networks, biology, and computer science.
In this paper, we consider influence maximization on the Ising model which is a prototypical example of undirected graphical models and has wide applications in many real-world problems.
We establish a sharp computational phase transition for influence maximization on sparse Ising models under a bounded budget: 
In the high-temperature regime, we give a linear-time algorithm for finding a small subset of variables and their values which achieve nearly optimal influence; 
In the low-temperature regime, we show that the influence maximization problem cannot be solved in polynomial time under commonly-believed complexity assumption.
The critical temperature coincides with the tree uniqueness/non-uniqueness threshold for Ising models which is also a critical point for other computational problems including approximate sampling and counting.


\end{abstract}

\section{Introduction}

Let $\mu$ be a distribution supported on $\{\pm 1\}^V$ where $V$ is a ground set of size $n$, and let $k \in \N^+$ be an integer corresponding to a budget.
We consider the following version of the influence maximization problem which asks to find a subset $S \subseteq V$ of size at most $k$ and a partial assignment $\sigma_S \in \{\pm 1\}^S$ which maximizes the expectation of $\sum_{v \in V} X_v$ conditioned on variables in $S$ receiving values specified by $\sigma_S$. 
In other words, we want to solve the following combinatorial optimization problem:
\begin{align}
\max_{\substack{S\subseteq V,\, |S| \le k \\ \sigma_S \in \{\pm 1\}^S}} 
\left\{ \E_\mu \left[ \sum_{v \in V} X_v \,\Bigg\vert\, X_S = \sigma_S \right] \right\}. \label{eq:inf-max}
\end{align}
Influence maximization is an important problem especially in the study of social networks and also has a vast number of applications in other areas \cite{KKT03,kempe2005influential}.

The problem of influence maximization has been extensively studied 
both theoretically and in more applied work. 
However, the families of distributions for which it was analyzed is somewhat limited.
The theoretical foundations for the model were introduced in~\cite{KKT03}
in terms of dynamical model where agents are infected if a function of their infected neighborhood surpasses a certain threshold. Algorithmic results and computational hardness are both stated in terms of properties of these threshold functions.
While the models introduced and analyzed in~\cite{KKT03,kempe2005influential,MR10} allow for and vastly generalize standard infections models, they do not apply to other standard models of correlated opinions. 

The main interest in our paper is in Ising models, 
which is one of the simplest and most popular graphical models for modeling the joint distribution of correlated discrete random variables. 
The Ising model was originally defined as statistical physics models, and nowadays they are widely used to model social networks, computer networks, and biological systems, see e.g.~\cite{LYS10,montanari2010spread,antonio2010ising,Majewski2001TheIM,lara2019analogy,lipowski2022ising}.

Consider a graph $G=(V,E)$, and let $\beta,h \in \R$. 
In the Ising distribution on $G$ parameterized by $\beta,h$, every configuration $\sigma \in \{\pm 1\}^V$ is assigned with a probability density
\[
\mu(\sigma) \propto \exp\left( \beta \sum_{uv \in E} \sigma_u \sigma_v + h \sum_{v \in V} \sigma_v \right).
\]
Here, $\beta$ is the inverse temperature describing the interaction between adjacent vertices. In particular if $\beta >0$ then neighboring vertices are more likely to receive the same value and the model is called \emph{ferromagnetic}; meanwhile, if $\beta < 0$ they could become repulsive to each other and the model is called \emph{antiferromagnetic}. 
The parameter $h$ is the external field of the system describing the bias of variables from outside. 
In general, every edge could have a distinct inverse temperature and every vertex a distinct external field; we refer to \cref{subsec:Ising} for this more general definition.

The problem of influence maximization for the Ising model was studied before in some special settings. 
It was shown by Bresler, Koehler, and Moitra~\cite{BrKoMo:19} using the GHS inequality that for the ferromagnetic Ising model, the influence maximization problem for the equilibrium measure is submodular and therefore obtains a $1-1/e$ approximation factor. 
More recently, a preprint by Chin, Moitra, Mossel and Sandon~\cite{chin2023power} shows that for very high temperature (small $\beta$) ferromagnetic Ising models with fixed parameter $\beta$, the influence maximization problem is approximately solved by the highest degree nodes. 
See also \cite{LYS10} for applications in social networks of influence maximization on Ising models.

While prior works provided some interesting algorithms for special cases, much remains unknown. First, many of the most natural models are not ferromagnetic. Second, we may be interested in an approximation factor better than $1-1/e$. 
Finally and importantly, we would like to understand the computational hardness of the problem. 

For Ising models, there exists a critical temperature $\beta_c$ which characterizes phase transitions of the model. 
Such a critical point $\beta_c$ depends on the maximum degree of the graph and is called the tree uniqueness/non-uniqueness threshold since it characterizes whether there exists a unique Gibbs measure for the Ising model on infinite regular trees. 
More importantly, the threshold $\beta_c$ pinpoints whether or not the model exhibits correlation decay \cite{Wei06,LLY13} or spectral independence \cite{ALO20,CLV20}, which are crucial properties for guaranteeing rapid mixing of natural Markov chains for sampling such as Glauber dynamics and polynomial-time algorithms for estimating the partition function.



We show that the critical temperature $\beta_c$ also pinpoints a computational phase transition for the influence maximization problem on sparse Ising models.
In fact, we consider a more general version of influence maximization where we want to maximize the influence on an arbitrary linear function of $X_v$'s under a bounded budget; see \cref{subsec:inf-max} for formal definitions. 
\begin{theorem}[Informal version of \cref{thm:max-inf,thm:hardness}]
Consider Ising models on bounded-degree graphs and let $k \in \N^+$ be a constant for the budget.
\begin{itemize}
\item If $|\beta| < \beta_c$, then one can find $(S,\sigma_S)$ whose influence is $\eps$-close to the optimal value within time $O(n) \cdot \poly(1/\eps)$;
\item If $|\beta| > \beta_c$, then there is no $\poly(n,1/\eps)$-time algorithm for influence maximization.
\end{itemize}

\end{theorem}

One important feature of our algorithmic result is that the running time of the algorithm is linear in $n$. 
Naively, one can easily obtain a polynomial-time algorithm by enumerating all possible $(S,\sigma_S)$ and find their corresponding influences. Since in the high-temperature regime we are able to approximately sample from the distribution or estimate the marginals in polynomial time, such a brute force algorithm runs in polynomial time; however, the exponent in $n$ is a large constant depending on $k$. 
Our algorithm has the advantage of being linear-time, assuming we have a constant budget $k = O(1)$.

To obtain a linear-time algorithm, we utilize the decay of correlation property and the spectral independence technique in a novel way. 
In the high-temperature regime (i.e.~$|\beta| < \beta_c$), the correlation/influence between a vertex $v$ and a subset $S \subseteq V$ of vertices is known to decay exponentially fast with their graph distance $\dist_G(v,S)$ \cite{Wei06,LLY13,CLV20}; see \cref{subsec:unique} for details. 
The key in our approach is to approximate the global influence of $(S,\sigma_S)$ on the whole vertex set $V$ by a local influence on only vertices sufficiently close to $S$ when correlation decay and spectral independence hold; see \cref{prop:approx-local}. 
The proof of the algorithmic result is provided in \cref{sec:alg}.

Meanwhile, in the low-temperature regime (i.e.~$|\beta| > \beta_c$) correlations or influences between two vertices can be non-vanishing even when their distance grows. For this reason simple Markov chain algorithms for sampling such as Glauber dynamics are known to be exponentially slow on such family, and our algorithmic approach fails for the same reason. In fact, it was known that approximate sampling and counting is $\mathsf{NP}$-hard in the antiferromagnetic case, i.e.~when $\beta < -\beta_c$ \cite{Sly10,SS14,GSV16}.
We establish hardness of influence maximization by giving a simple reduction from approximating the partition function of Ising models. The proof can be found in \cref{sec:hardness}.

\section{Preliminaries}


Suppose $G=(V,E)$ is a graph.
For two vertices $u,v \in V$, let $\dist_G(u,v)$ denote their graph distance in $G$. 
For any $u \in V$ and any $r>0$, let $\ball(u,r) = \{v\in V: \dist_G(u,v) \le r\}$ be the ball of radius $r$ around $u$. 
Further, for any $S \subseteq V$ let $\ball(S,r) = \bigcup_{u \in S} \ball(u,r)$. 

For $r > 0$, let $G^{\le r}$ denote the graph with the same vertex set $V$ and two vertices $u,v$ adjacent iff $\dist_G(u,v) \le r$. 
For $S\subseteq V$, let $G[S]$ be the subgraph induced on $S$.

\subsection{Ising model}
\label{subsec:Ising}

Suppose $G=(V,E)$ is a graph. 
Let $\beta \in \R^E$ be a vector of edge couplings 
and $h \in \R^V$ be a vector of external fields. 
The Gibbs distribution $\mu = \mu_{G,\beta,h}$ of the Ising model $(G,\beta,h)$ is given by
\begin{equation}\label{eq:Gibbs}
\mu(\sigma) := \frac{1}{Z} \exp\left( \sum_{uv \in E} \beta_{uv} \sigma_u \sigma_v + \sum_{v \in V} h_v \sigma_v \right), \qquad \forall \sigma \in \{\pm 1\}^V
\end{equation}
where the partition function $Z = Z_{G,\beta,h}$ is defined by
\[
Z = \sum_{\sigma \in \{\pm 1\}^V} \exp\left( \sum_{uv \in E} \beta_{uv} \sigma_u \sigma_v + \sum_{v \in V} h_v \sigma_v \right).
\]

For an integer $\Delta \ge 3$ and a real $\gamma>0$, let $\MM(\Delta,\gamma)$ be the family of all Ising models $(G,\beta,h)$ satisfying:
\begin{enumerate}
\item The graph $G$ has maximum degree at most $\Delta$;
\item For all $uv \in E$ it holds $(\Delta-1) (\tanh|\beta_{uv}|) \le \gamma$.
\end{enumerate}

We remark that for the family $\MM(\Delta,\gamma)$ every edge coupling can be either ferromagnetic (i.e.~$\beta_{uv} > 0$) or antiferromagnetic (i.e.~$\beta_{uv} < 0$).

The critical temperature is given by $\beta_c(\Delta) = \mathrm{arctanh}(1/(\Delta-1))$.
Hence, for any Ising model from the family $\MM(\Delta,\gamma)$ where $\gamma < 1$, every edge coupling satisfies $|\beta_{uv}| < \beta_c(\Delta)$.

\subsection{Tree uniqueness, strong spatial mixing, total influence decay}
\label{subsec:unique}

In the high-temperature regime, strong spatial mixing (correlation decay) was known for the family $\MM(\Delta,\gamma)$ for any $\gamma<1$ \cite{Wei06,SST14,LLY13}. 
Recently, \cite{CLV20} established $\ell_\infty$-spectral independence by showing the exponential decay of total influences via Weitz's self-avoiding tree approach \cite{Wei06}.

\begin{lemma}[\cite{Wei06,SST14,LLY13,CLV20}]
\label{lem:SSM-TID}
For any $\Delta \ge 3$ and $\delta \in (0,1)$, there exists a constant $C = C(\Delta,\delta) > 0$ such that the following holds.
Consider an Ising model on a graph $G=(V,E)$ from the family $\MM(\Delta,1-\delta)$. 
Let $\Lambda \subseteq V$ and $\tau \in \{\pm 1\}^\Lambda$ be an arbitrary pinning.
\begin{itemize}
\item (Strong Spatial Mixing) For any $u \in V \setminus \Lambda$ and $L \in \N^+$, for any subset $W \subseteq V \setminus \Lambda \setminus\{u\}$ such that $\dist_G(u,W) \ge L$ and any two spin assignments $\sigma_W,\xi_W \in \{\pm 1\}^W$, we have
\[
\left| \Pr_{\mu^\tau}\left( X_u = + \mid X_W = \sigma_W \right) - \Pr_{\mu^\tau}\left( X_u = + \mid X_W = \xi_W \right) \right| 
\le C(1-\delta)^L.
\] 

\item (Total Influence Decay) For any $u \in V \setminus \Lambda$ and $L \in \N^+$, we have
\[
\sum_{v \in V \setminus \Lambda: \, \dist_G(u,v) \ge L} \left| \Pr_{\mu^\tau}(X_v = + \mid X_u = +) - \Pr_{\mu^\tau}(X_v = + \mid X_u = -) \right| 
\le C(1-\delta)^L.
\]
\end{itemize}
\end{lemma}

\subsection{Influence maximization}
\label{subsec:inf-max}

Consider the Ising model on a graph $G = (V,E)$ with edge couplings $\beta \in \R^E$ and external fields $h \in \R^V$.
Let $a \in \R^V$ be a vector of vertex weights. 

\begin{definition}[Global Influence]
For a subset $S \subseteq V$ of vertices and a partial assignment $\sigma_S \in \{\pm 1\}^S$ on $S$, define the \emph{global} influence of $(S,\sigma_S)$ on the linear function $a \cdot X$ to be
\[
\Phi_{G,\beta,h,a}(S,\sigma_S) = \E \left[ \sum_{v \in V} a_v X_v \,\Bigg\vert\, X_S = \sigma_S \right] - \E \left[ \sum_{v \in V} a_v X_v \right]
\]
where $X \in \{\pm 1\}^V$ is sampled from the Ising model $(G,\beta,h)$.
\end{definition}


Let $k \in \N^+$ be an integer representing the budget.
In this paper, we consider the {$k$-\infmax} problem where we want to select a subset $S \subseteq V$ of size at most $k$ and a partial assignment $\sigma_S \in \{\pm 1\}^S$ which achieves almost the maximum global influence.
Formally, the problem {$k$-\infmax} is defined as follows.

\medskip
\noindent
$k$-\infmax

\noindent
\textbf{Input}: $(G,\beta,h)$ an Ising model; $a \in \R^V$ a vector of vertex weights; $\eps > 0$ an error parameter.

\noindent
\textbf{Output}: a subset $\hat{S} \subseteq V$ with $|\hat{S}| \le k$ and a partial assignment $\sigma_{\hat{S}} \in \{\pm 1\}^{\hat{S}}$ such that
\[
\Phi_{G,\beta,h,a}(\hat{S},\sigma_{\hat{S}}) 
\ge \max_{\substack{S\subseteq V,\, |S| \le k \\ \sigma_S \in \{\pm 1\}^S}} \left\{ \Phi_{G,\beta,h,a}(S,\sigma_S) \right\} - \eps.
\]

We say a weight vector $a \in \R^V$ is $L$-bounded if $\norm{a}_\infty \le L$, i.e., $|a_v| \le L$ for all $v \in V$.
We are interested in maximizing the global influence for bounded weights.
Since $\Phi_{G,\beta,h,ta}(S,\sigma_S) = t \cdot \Phi_{G,\beta,h,a}(S,\sigma_S)$, we may assume that $a$ is $1$-bounded.
Furthermore, we consider influence maximization with a constant budget, namely $k=O(1)$, which is already interesting and captures many real-world settings.
Our goal is to find an algorithm for {$k$-\infmax} with running time polynomial in $n$ and $1/\eps$, and understand the computational complexity of it.

\subsection{Main results}


\begin{theorem}[Algorithmic Result]
\label{thm:max-inf}
Suppose $\Delta \ge 3$ is an integer and $\delta \in (0,1)$ is a real. 
For any integer $k \in \N^+$, there exists a deterministic algorithm that solves {$k$-\infmax} 
for the family $\MM(\Delta,1-\delta)$ and $1$-bounded vertex weights 
with running time $O(n) \cdot (1/\eps)^{O(1)}$.
\end{theorem}

\begin{theorem}[Hardness Result]
\label{thm:hardness}
Suppose $\Delta \ge 3$ is an integer and $\delta > 0$ is a real. 
For any integer $k \in \N^+$, there is no randomized algorithm that can solve $k$-{\infmax} 
for the family $\MM(\Delta,1+\delta)$ and $1$-bounded vertex weights with probability at least $3/4$ in time $\poly(n,1/\eps)$, assuming $\mathsf{RP}\neq \mathsf{NP}$.
\end{theorem}

\section{Algorithmic Result}
\label{sec:alg}

We prove our algorithmic result \cref{thm:max-inf} by localizing the global influence of a subset $S$ of vertices to a ball around $S$.
For high-temperature Ising models, such local influence approximates the global influence effectively. 
Furthermore, one can approximately maximize the local influence by the local nature of the problem. 
Together this gives an approximation algorithm for global influence maximization.

To begin we first define the notion of local influences.

\begin{definition}[Local Influence]
Let $r \in \N^+$. For a subset $S \subseteq V$ of vertices and a partial assignment $\sigma_S \in \{\pm 1\}^S$ on $S$, define the \emph{local} influence of $(S,\sigma_S)$ on the linear function $a \cdot X$ to be
\[
\Phi_{G,\beta,h,a}^{(r)}(S,\sigma_S) = \E_{G[\ball(S,r)]} \left[ \sum_{v \in \ball(S,r)} a_v X_v \,\Bigg\vert\, X_S = \sigma_S \right] - \E_{G[\ball(S,r)]} \left[ \sum_{v \in \ball(S,r)} a_v X_v \right],
\]  
where $X \in \{\pm 1\}^V$ is sampled from the Ising model on the induced subgraph $G[\ball(S,r)]$ with $\beta,h$ restricted on it.
\end{definition}

Notice that, if $S=\{u,w\}$ and $\dist_G(u,w) > 2r+1$, then the induced subgraph $G[\ball(\{u,w\},r)]$ is the disjoint union of $G[\ball(u,r)]$ and $G[\ball(w,r)]$, and we can further decompose the local influence as
\begin{align*}
\Phi_{G,\beta,h,a}^{(r)}(S,\sigma_S) 
={}& \left( \E_{G[\ball(u,r)]} \left[ \sum_{v \in \ball(u,r)} a_v X_v \,\Bigg\vert\, X_u = \sigma_u \right] - \E_{G[\ball(u,r)]} \left[ \sum_{v \in \ball(u,r)} a_v X_v \right] \right) \\
&+ \left( \E_{G[\ball(w,r)]} \left[ \sum_{v \in \ball(w,r)} a_v X_v \,\Bigg\vert\, X_w = \sigma_w \right] - \E_{G[\ball(w,r)]} \left[ \sum_{v \in \ball(w,r)} a_v X_v \right] \right) \\
={}& \Phi_{G,\beta,h,a}^{(r)}(u,\sigma_u) + \Phi_{G,\beta,h,a}^{(r)}(w,\sigma_w).  
\end{align*}
Thus, we see that for local influence we are able to decompose it into clusters of vertices close to each other; more specifically the clusters are connected components of the induced subgraph $G[\ball(S,r)]$, see \cref{lem:decomp} for a precise statement.

We now present two main propositions for establishing \cref{thm:max-inf}. 
Fix $\Delta \ge 3$, $\delta \in (0,1)$ and $k \in \N^+$. In the propositions below $O(\cdot) = O_{\Delta,\delta,k}(\cdot)$ hides a constant depending on $\Delta,\delta,k$.

We first show that for high-temperature Ising models, the global influence is well-approximated by the local influence for sufficiently large radius $r \in \N^+$. 

\begin{proposition}\label{prop:approx-local}
Consider an Ising model on a graph $G=(V,E)$ from the family $\MM(\Delta,1-\delta)$ and a $1$-bounded weight vector $a \in \R^V$. 
For any $\eps > 0$, there exists $r = O(\log(1/\eps))$ such that for all $S \subseteq V$ with $|S| \le k$ and all $\sigma_S \in \{\pm 1\}^S$, we have
\begin{align}\label{eq:approx-local}
\left|\Phi_{G,\beta,h,a}(S,\sigma_S) - \Phi_{G,\beta,h,a}^{(r)}(S,\sigma_S) \right| \le \eps.
\end{align}
\end{proposition}

Next, we give a linear-time algorithm for approximately maximizing the local influence.

\begin{proposition}\label{prop:alg-local}
Consider an Ising model on a graph $G=(V,E)$ from the family $\MM(\Delta,1-\delta)$ and a $1$-bounded weight vector $a \in \R^V$. 
For any $\eps > 0$ and $r \in \N^+$, there exists an algorithm that finds a subset $\hat{S} \subseteq V$ with $|\hat{S}| \le k$ and a partial assignment $\sigma_{\hat{S}} \in \{\pm 1\}^{\hat{S}}$ such that
\begin{align}\label{eq:alg-local}
\Phi_{G,\beta,h,a}^{(r)}(\hat{S},\sigma_{\hat{S}}) 
\ge \max_{\substack{S\subseteq V,\, |S| \le k \\ \sigma_S \in \{\pm 1\}^S}} \left\{ \Phi_{G,\beta,h,a}^{(r)}(S,\sigma_S) \right\} - \eps.
\end{align}
The running time of the algorithm is $O(n) \cdot (1/\eps)^{O(1)} \cdot e^{O(r)}$.
\end{proposition}

\cref{thm:max-inf} follows immediately from \cref{prop:approx-local,prop:alg-local}. 
For ease of notations we omit $G,\beta,h,a$ in the subscripts for the rest of the paper when it is clear from the context.

\begin{proof}[Proof of \cref{thm:max-inf}]
Define the optimal solutions
\begin{align}
(S^*,\sigma_{S^*}) 
&= \argmax_{\substack{S\subseteq V,\, |S| \le k \\ \sigma_S \in \{\pm 1\}^S}} \left\{ \Phi(S,\sigma_S) \right\} \\
\text{and}\qquad
(S^\dagger,\sigma_{S^\dagger}) 
&= \argmax_{\substack{S\subseteq V,\, |S| \le k \\ \sigma_S \in \{\pm 1\}^S}} \left\{ \Phi^{(r)}(S,\sigma_S) \right\}. \label{eq:opt}
\end{align}
Let $r = O(\log(1/\eps))$ be from \cref{prop:approx-local} such that \cref{eq:approx-local} holds with the error in the right-hand side being $\eps/3$. 
(Note that we can compute $r$ efficiently by \cref{eq:rho,eq:r} from the proof of \cref{prop:approx-local}.)
For this $r$ use the algorithm from \cref{prop:alg-local} to find $\hat{S} \subseteq V$ and $\sigma_{\hat{S}} \in \{\pm 1\}^{\hat{S}}$ such that \cref{eq:alg-local} holds with the error in the right-hand side being $\eps/3$.
Thus, we conclude that
\begin{align*}
\Phi(\hat{S},\sigma_{\hat{S}}) 
&\overset{ \text{\cref{eq:approx-local}} }{\ge} \Phi^{(r)}(\hat{S},\sigma_{\hat{S}}) - \frac{\eps}{3}
\overset{ \text{\cref{eq:alg-local}} }{\ge} \Phi^{(r)}(S^\dagger,\sigma_{S^\dagger}) - \frac{2\eps}{3} \\
&\overset{ \text{\cref{eq:opt}} }{\ge} \Phi^{(r)}(S^*,\sigma_{S^*}) - \frac{2\eps}{3}
\overset{ \text{\cref{eq:approx-local}} }{\ge} \Phi(S^*,\sigma_{S^*}) - \eps
\end{align*}
as wanted.
The running time of the algorithm is $O(n) \cdot (1/\eps)^{O(1)} \cdot e^{O(r)} = O(n) \cdot (1/\eps)^{O(1)}$.
\end{proof}

\subsection{Proof of \texorpdfstring{\cref{prop:approx-local}}{Proposition 3.2}}

Fix $S \subseteq V$ and $\sigma_S \in \{\pm 1\}^S$,
and define
\begin{align*}
\f(k,\ell) 
&= \E_{G[\ball(S,k)]} \left[ \sum_{v \in \ball(S,\ell)} a_v X_v \,\Bigg\vert\, X_S = \sigma_S \right] - \E_{G[\ball(S,k)]} \left[ \sum_{v \in \ball(S,\ell)} a_v X_v \right] \\
&= \sum_{v \in \ball(S,\ell)} a_v \left( \E_{G[\ball(S,k)]} \left[ X_v \,\vert\, X_S = \sigma_S \right] - \E_{G[\ball(S,k)]} \left[ X_v \right] \right).
\end{align*}
Define $\ball(v,\infty)$ to be the connected component containing $v$ and $\ball(S,\infty) = \bigcup_{v \in S} \ball(v,\infty)$. 
Then we have $\f(\infty,\infty) = \Phi(S,\sigma_S)$ and $\f(r,r) = \Phi^{(r)}(S,\sigma_S)$; 
to see the former, observe that
\begin{align*}
\Phi(S,\sigma_S) &= \E_G \left[ \sum_{v \in V} a_v X_v \,\Bigg\vert\, X_S = \sigma_S \right] - \E_G \left[ \sum_{v \in V} a_v X_v \right] \\
&= \E_G \left[ \sum_{v \in \ball(S,\infty)} a_v X_v \,\Bigg\vert\, X_S = \sigma_S \right] - \E_G \left[ \sum_{v \in \ball(S,\infty)} a_v X_v \right] \tag{$X_S$ and $X_{V \setminus \ball(S,\infty)}$ are independent} \\
&= \E_{G[\ball(S,\infty)]} \left[ \sum_{v \in \ball(S,\infty)} a_v X_v \,\Bigg\vert\, X_S = \sigma_S \right] - \E_{G[\ball(S,\infty)]} \left[ \sum_{v \in \ball(S,\infty)} a_v X_v \right] \tag{$\mu_G = \mu_{G[\ball(S,\infty)]} \otimes \mu_{G[V \setminus \ball(S,\infty)]}$} \\
&= \f(\infty,\infty).
\end{align*}

Therefore, it suffices to show that $|\f(\infty,\infty) - \f(r,r)| \le \eps$, which follows immediately from the following three lemmas.

\begin{lemma}\label{lem:infty-rho}
There exists $\rho = O(\log(1/\eps))$ such that $|\f(\infty,\infty) - \f(\infty,\rho)| \le \eps/3$.
\end{lemma}

\begin{proof}
We define
\begin{align}\label{eq:rho}
\rho = \ceil{ \frac{1}{\delta} \log\left( \frac{6Ck}{\eps} \right) }.
\end{align}
For simplicity we write $\Pr = \Pr_{G[\ball(S,\infty)]}$ for the Ising distribution $\mu_{G[\ball(S,\infty)]}$, and $\E = \E_{G[\ball(S,\infty)]}$ for the expectation over $\mu_{G[\ball(S,\infty)]}$. 
By definitions we have that
\begin{align*}
|\f(\infty,\infty) - \f(\infty,\rho)| 
= \left| \sum_{v \in \ball(S,\infty) \setminus \ball(S,\rho)} a_v \left( \E \left[ X_v \,\vert\, X_S = \sigma_S \right] - \E \left[ X_v \right] \right) \right|. 
\end{align*}
Suppose $S = \{v_1,\dots,v_{k'}\}$ where $k'=|S| \le k$. For $0\le i \le k'$ we define $S_i = \{v_1,\dots,v_i\}$ and let $\sigma_{S_i}$ be $\sigma_S$ restricted to $S_i$. 
Then it follows that
\begin{align*}
& |\f(\infty,\infty) - \f(\infty,\rho)| \\
={}& \left| \sum_{i=1}^{k'} \sum_{v \in \ball(S,\infty) \setminus \ball(S,\rho)} a_v \left( \E \left[ X_v \,\vert\, X_{S_i} = \sigma_{S_i} \right] - \E \left[ X_v \,\vert\, X_{S_{i-1}} = \sigma_{S_{i-1}} \right] \right) \right| \\
\overset{\text{(i)}}{\le}{}& \sum_{i=1}^{k'} \sum_{v \in \ball(S,\infty) \setminus \ball(S,\rho)} |a_v| \cdot \left| \E \left[ X_v \,\vert\, X_{S_i} = \sigma_{S_i} \right] - \E \left[ X_v \,\vert\, X_{S_{i-1}} = \sigma_{S_{i-1}} \right] \right| \\
\overset{\text{(ii)}}{\le}{}& \sum_{i=1}^{k'} \sum_{v \in \ball(S,\infty) \setminus \ball(S,\rho)} 
2 \left| \Pr^{\sigma_{S_{i-1}}}(X_v = +1 \mid X_{v_i} = +1) - \Pr^{\sigma_{S_{i-1}}}(X_v = +1 \mid X_{v_i} = -1) \right| \\
\overset{\text{(iii)}}{\le}{}& 2Ck(1-\delta)^\rho \\
\overset{\text{(iv)}}{\le}{}& \frac{\eps}{3},
\end{align*}
where (i) is the triangle inequality, (ii) follows from $|a_v| \le 1$ and expanding the expectation, 
(iii) follows from Total Influence Decay (\cref{lem:SSM-TID}), and (iv) is by our choice of $\rho$. 
\end{proof}

\begin{lemma}\label{lem:infty-r}
Given $\rho \in \N^+$, there exists $\rho < r = O(\rho+\log(1/\eps))$ such that $|\f(\infty,\rho) - \f(r,\rho)| \le \eps/3$. 
\end{lemma}

\begin{proof}
We define
\begin{align}\label{eq:r}
r = \rho + \ceil{ \frac{1}{\delta} \left( \log \left(\frac{24C}{\eps}\right) + \rho \log \Delta \right) }.
\end{align}
By the triangle inequality and $\norm{a}_\infty \le 1$ we have that
\begin{align*}
&|\f(\infty,\rho) - \f(r,\rho)| \\
\le{}& \sum_{v \in \ball(S,\rho)} \left| \E_{G_\infty} \left[ X_v \,\vert\, X_S = \sigma_S \right] - \E_{G_r} \left[ X_v \,\vert\, X_S = \sigma_S \right] \right|
+ \left| \E_{G_\infty} \left[ X_v \right] - \E_{G_r} \left[ X_v \right] \right| \\
={}& \sum_{v \in \ball(S,\rho)} 2 \left| \Pr_{G_\infty}^{\sigma_S} \left( X_v = +1 \right) - \Pr_{G_r}^{\sigma_S} \left( X_v = +1 \right) \right|
+ 2 \left| \Pr_{G_\infty} \left( X_v = +1 \right) - \Pr_{G_r} \left( X_v = +1 \right) \right|
\end{align*}
where $G_\infty = G[\ball(S,\infty)]$, $G_r = G[\ball(S,r)]$, $\Pr_{G_\infty}^\tau = \Pr_{\mu_{G_\infty}^\tau}$, and $\Pr_{G_r}^\tau = \Pr_{\mu_{G_r}^\tau}$. 
Let $U = \{u \in V: \dist_G(u,S) = r\}$.
For any $v \in \ball(S,\rho)$, we can couple $X_v \sim \Pr_{G_\infty}(X_v=\cdot)$ and $X'_v \sim \Pr_{G_r}(X'_v=\cdot)$ by first revealing the spin assignments $X_U \sim \Pr_{G_\infty}(X_U=\cdot)$ and $X'_U \sim \Pr_{G_r}(X'_U=\cdot)$ on $U$ independently and then couple $X_v,X'_v$ optimally conditioned on $X_U,X'_U$ respectively. 
Therefore, we deduce that
\begin{align*}
\left| \Pr_{G_\infty}^{\sigma_S} \left( X_v = +1 \right) - \Pr_{G_r}^{\sigma_S} \left( X_v = +1 \right) \right| 
&\le \max_{\sigma_U,\tau_U \in \{\pm 1\}^U} \left| \Pr_{G}^{\sigma_S,\sigma_U} \left( X_v = +1 \right) - \Pr_{G}^{\sigma_S,\tau_U} \left( X_v = +1 \right) \right| \\
&\le C(1-\delta)^{r-\rho},
\end{align*}
where the first inequality follows from the coupling procedure and the fact that 
$$ \Pr_{G_\infty}^{\sigma_S,\sigma_U}(X_v=\cdot) = \Pr_{G}^{\sigma_S,\sigma_U}(X_v=\cdot) = \Pr_{G_r}^{\sigma_S,\sigma_U}(X_v=\cdot), $$ 
and the second inequality follows from Strong Spatial Mixing (\cref{lem:SSM-TID}) and $\dist_G(v,U) \ge \dist_G(S,U) - \dist_G(v,S) \ge r-\rho$. 
Similarly, we also have 
\[
\left| \Pr_{G_\infty} \left( X_v = +1 \right) - \Pr_{G_r} \left( X_v = +1 \right) \right| 
\le C(1-\delta)^{r-\rho}.
\]

Hence, combining everything above and $|\ball(S,\rho)| \le 2\Delta^\rho$ we get
\[
|\f(\infty,\rho) - \f(r,\rho)| 
\le 2\Delta^\rho \cdot 4C(1-\delta)^{r-\rho} 
\le \frac{\eps}{3},
\]
as wanted.
\end{proof}

\begin{lemma}\label{lem:rho-r}
For $\rho \in \N^+$ in \cref{eq:rho} and any integer $r \ge \rho$, we have $|\f(r,\rho) - \f(r,r)| \le \eps/3$.
\end{lemma}

\begin{proof}
The proof is exactly the same as for \cref{lem:infty-rho}; one only needs to replace $G[\ball(S,\infty)]$ with $G[\ball(S,r)]$ once noting that Total Influence Decay still holds on any subgraph.
\end{proof}

\begin{proof}[Proof of \cref{prop:approx-local}]
For $\rho,r$ given in \cref{eq:rho,eq:r}, we obtain from the triangle inequality and \cref{lem:infty-rho,lem:infty-r,lem:rho-r} that
\[
|\f(\infty,\infty) - \f(r,r)| \le
|\f(\infty,\infty) - \f(\infty,\rho)| + |\f(\infty,\rho) - \f(r,\rho)| + |\f(r,\rho) - \f(r,r)| \le \eps,
\]
as claimed.
\end{proof}

\subsection{Proof of \texorpdfstring{\cref{prop:alg-local}}{Proposition 3.3}}

For a graph $G$, let $\mathsf{cc}(G)$ denote the set of all connected components of $G$, where each connected component is viewed as a subset of vertices. 
The following decomposition lemma is easy to verify.
\begin{lemma}\label{lem:decomp}
For any $r \in \N$, $S \subseteq V$, and $\sigma_S \in \{\pm 1\}^S$, we have
\[
\Phi^{(r)}(S,\sigma_S) = \sum_{T \in \mathsf{cc} (G^{\le 2r+1}[S])} \Phi^{(r)}(T,\sigma_T).
\]
\end{lemma}

\begin{proof}
Follows from the fact that 
\[
\mu_{G[\ball(S,r)]} = \bigotimes_{T \in \mathsf{cc} (G^{\le 2r+1}[S])} \mu_{G[\ball(T,r)]},
\]
and the same for the conditional distribution with a partial assignment $\sigma_S$ on $S$.
\end{proof}

Hence, it suffices to consider all local influences for subsets of vertices that are connected in $G^{\le 2r+1}$. 
Our algorithm is given below.

\medskip
\noindent\textbf{Outline of the Algorithm}
\begin{enumerate}[Step 1.]

\item Construct a graph $H=(V_H,E_H)$ as follows. 
\begin{enumerate}
\item[(1a)] The vertex set $V_H$ consists of all non-empty subsets $T \subseteq V$ of vertices of size at most $k$ such that $G^{\le 2r+1}[T]$ is connected. 
\item[(1b)] Two distinct subsets $T_1,T_2$ are adjacent iff $G^{\le 2r+1}[T_1 \cup T_2]$ is connected; equivalently, $T_1,T_2$ are non-adjacent iff $\dist_G(T_1,T_2) > 2r+1$. 
\end{enumerate}
\emph{$\vartriangleright$ (\cref{lem:step1}) We can construct $H$ in $O(n) \cdot e^{O(r)}$ time.}

\item Each vertex $T \in V_H$ is assigned an integral cost $c_T \in \N^+$, a real weight $w_T \in \R$, and a partial assignment $\xi_T \in \{\pm 1\}^T$ as follows.
\begin{enumerate}
\item[(2a)] The cost of $T$ is its size; i.e., $c_T = |T|$.
\item[(2b)] 
For every $\sigma_T \in \{\pm 1\}^T$, compute $\psi_T(\sigma_T)$ such that
\begin{align}\label{eq:psi-approx}
\left| \psi_T(\sigma_T) - \Phi^{(r)}(T,\sigma_T) \right| \le \frac{\eps}{2k}.
\end{align}
The weight of $T$ is the maximum value of $\psi_T(\sigma_T)$ and the associated partial assignment is the maximizer: 
\begin{align}
\xi_T &= \argmax_{\sigma_T \in \{\pm 1\}^T} \psi_T(\sigma_T) \label{eq:xiT} \\
\text{and}\qquad
w_T &= \psi_T(\xi_T) = \max_{\sigma_T \in \{\pm 1\}^T} \psi_T(\sigma_T). \label{eq:wT}
\end{align}
\end{enumerate}
\emph{$\vartriangleright$ (\cref{lem:step2}) For each $T \in V_H$, we can compute $c_T$, $w_T$, and $\xi_T$ in $(1/\eps)^{O(1)} \cdot e^{O(r)}$ time.}

\item Given the graph $H$, costs $\{c_T\}_{T \in V_H}$, and weights $\{w_T\}_{T \in V_H}$, find a maximum weighted independent set $I^*$ of $H$ with total cost at most $k$; namely, 
\begin{align}
\max \quad&\sum_{T \in I} w_T \label{eq:I*-max} \\ 
\mathrm{s.t.} \quad & \text{$I$ is an independent set of $H$}; \nonumber \\
& \sum_{T \in I} c_T \le k. \nonumber
\end{align}
\emph{$\vartriangleright$ (\cref{lem:step3}) We can find $I^*$ in $O(n) \cdot e^{O(r)}$ time.}

\item Output 
\[
\hat{S} = \bigcup_{T \in I^*} T
\qquad\text{and}\qquad
\sigma_{\hat{S}} = (\xi_T)_{T \in I^*}.
\]
\emph{$\vartriangleright$ (\cref{lem:step4}) We have $\Phi^{(r)}(\hat{S},\sigma_{\hat{S}}) \ge \Phi^{(r)}(S^\dagger,\sigma_{S^\dagger}) -\eps$ as desired.}
\end{enumerate}
\medskip

We show the correctness of our algorithm and analyze the running time of it in the following sequence of lemmas.
Throughout, we assume that $G$ has maximum degree at most $\Delta \ge 3$, the Ising model on $G$ is from the family $\MM(\Delta,1-\delta)$, and $k\in \N^+$ is fixed.

Both Step 1 and 2 can be completed in $O(n) \cdot (1/\eps)^{O(1)} \cdot e^{O(r)}$ time.
\begin{lemma}[Step 1]
\label{lem:step1}
The graph $H$ has $N = O(n) \cdot e^{O(r)}$ vertices and maximum degree $D = e^{O(r)}$. Furthermore, one can construct $H$ in $O(n) \cdot e^{O(r)}$ time. 
\end{lemma}

\begin{proof}
Follows from results in \cite{PR17,BCKL13} for counting and enumerating bounded-size connected induced subgraphs in a bounded-degree graph. 
\end{proof}

\begin{lemma}[Step 2]
\label{lem:step2}
For each $T \in V_H$, its cost $c_T$, weight $w_T$, and partial assignment $\xi_T$ can be computed in $e^{O(r)} \cdot (1/\eps)^{O(1)}$ time.
\end{lemma}

\begin{proof}
The cost $c_T$ is trivial.
For the weight $w_T$, we first compute $\psi_T(\sigma_T)$ which approximates $\Phi^{(r)}(T,\sigma_T)$, for all choices of $\sigma_T \in \{\pm 1\}^T$. 
The approximation of $\Phi^{(r)}(T,\sigma_T)$
follows from \cite{Wei06,SST14,Bar16book,PR17} which present deterministic approximate counting algorithms ($\mathsf{FPTAS}$) for high-temperature Ising models.
More specifically, observe that by definition $\Phi^{(r)}(T,\sigma_T)$ is a linear combination of marginal probabilities at each vertex in $\ball(T,r)$ either with or without the pinning $\sigma_T$.
Thus, one can estimate all such marginals within an additive error $\eps' = \frac{\eps}{4k|\ball(T,r)|}$, and then obtain $\psi_T(\sigma_T)$ from these estimates such that 
$$ \left| \psi_T(\sigma_T) - \Phi^{(r)}(T,\sigma_T) \right| \le 2|\ball(T,r)| \cdot \eps' = \frac{\eps}{2k}. $$
The running time of this is 
$$ |\ball(T,r)|^{O(1)} \cdot (1/\eps')^{O(1)} = e^{O(r)} \cdot (1/\eps)^{O(1)}. $$
Given $\psi_T(\sigma_T)$ for all $\sigma_T \in \{\pm 1\}^T$, we can then find $\xi_T$ and $w_T$. Note that the number of choices of $\sigma_T$ is at most $2^k$ which is $O(1)$.
\end{proof}

The algorithm for Step 3 is given by the following lemma.

\begin{lemma}[Step 3]
\label{lem:step3}
Let $H=(V,E)$ be an $N$-vertex graph of maximum degree at most $D \ge 3$. Suppose every vertex $T \in V$ is assigned an integral cost $c_T \in \N^+$ and a real weight $w_T \in \R$. Then for any fixed $k \in \N^+$ with $k = O(1)$, there exists an algorithm that finds a maximum weighted independent set $I$ of $H$ with total cost at most $k$ in time $O(DN) + D^{O(1)}$. 
\end{lemma}

\begin{proof}
For $1\le i \le k$, define $V^{(i)} = \{T \in V: c_T = i\}$ to be the set of all vertices of cost $i$. 
Let $U^{(i)} = \{T^{(i)}_1,\dots,T^{(i)}_{t_i}\} \subseteq V^{(i)}$ be the $t_i$ vertices of largest weights from $V^{(i)}$ (break ties arbitrarily), where $t_i = |U^{(i)}| = \min\{k(D+1), |V^{(i)}|\}$.
Finally, let $U = \bigcup_{i=1}^k U^{(i)}$. Observe that $U$ can be found in $O(DN)$ time.

We claim that there exists a maximum weighted independent set $I^*$ with total cost at most $k$ such that $I^*$ is completely contained in $U$. To prove the claim, let us define $I^*$ to be the maximum weighted independent set with total cost at most $k$ that contains the most vertices in $U$, and it suffices to prove $I^* \subseteq U$. 
Suppose for sake of contradiction that $I^* \not\subseteq U$. 
Take any $T \in I^* \setminus U$, and assume that $T \in V^{(i)}$ for some $i$.
Since $T \not\in U = \bigcup_{j=1}^k U^{(j)}$, we have $T \in V^{(i)} \setminus U^{(i)}$ 
and it holds $|V^{(i)}| > |U^{(i)}| = k(D+1)$. 
We say a vertex $T$ blocks a vertex $T'$ if either $T = T'$ or $T,T'$ are adjacent. 
Thus, every vertex blocks at most $D+1$ vertices. 
It follows that vertices in $I^* \setminus \{T\}$ block at most $(D+1)(k - 1) \le |U^{(i)}|-1$ vertices altogether. 
Hence, there exists a vertex $T' \in U^{(i)}$ which is not blocked by $I^* \setminus \{T\}$.
In particular, $I' = I^* \setminus \{T\} \cup\{T'\}$ is an independent set with the same cost ($c_{T'} = i = c_T$) and no smaller weight ($w_{T'} \ge w_T$ by the definition of $U^{(i)}$), while containing one more vertex from $U$. 
This is a contradiction. 

Given the claim, one only needs to enumerate all subsets of $U$ of size at most $k$ to find the maximum weighted independent set with cost constraint $k$. Since $|U| = O(D)$, this can be done in $D^{O(1)}$ time, finishing the proof.
\end{proof}

Finally, we show the correctness of our algorithm.
\begin{lemma}[Step 4]
\label{lem:step4}
Let $(S^\dagger,\sigma_{S^\dagger})$ be the maximizer for the local influence defined in \cref{eq:opt}.
We have that
\[
\Phi^{(r)}(\hat{S},\sigma_{\hat{S}}) \ge \Phi^{(r)}(S^\dagger,\sigma_{S^\dagger}) -\eps.
\]
\end{lemma}

\begin{proof}
By \cref{lem:decomp}, the optimal solution $(S^\dagger,\sigma_{S^\dagger})$ corresponds to an independent set $I^\dagger$ of $H$ such that
\[
S^\dagger = \bigcup_{T \in I^\dagger} T
\qquad\text{and}\qquad
\sigma_{S^\dagger} = \left( \eta_T \right)_{T \in I^\dagger}
~\text{where}~
\eta_{T} = \argmax_{\sigma_T \in \{\pm 1\}^T} \left\{ \Phi^{(r)}(T,\sigma_T) \right\}.
\]
We then deduce that
\begin{align*}
\Phi^{(r)}(\hat{S},\sigma_{\hat{S}})
&\overset{\text{Lem \ref{lem:decomp}}}{=} \sum_{T \in I^*} \Phi^{(r)}(T,\xi_T) 
\overset{\text{\cref{eq:psi-approx}}}{\ge} \sum_{T \in I^*} \psi_T(\xi_T) - \frac{\eps}{2} 
\overset{\text{\cref{eq:wT}}}{=} \sum_{T \in I^*} w_T - \frac{\eps}{2} \\
&\overset{\text{\cref{eq:I*-max}}}{\ge} \sum_{T \in I^\dagger} w_T - \frac{\eps}{2} 
\overset{\text{\cref{eq:wT}}}{=} \sum_{T \in I^\dagger} \psi_T(\xi_T) - \frac{\eps}{2} 
\overset{\text{\cref{eq:xiT}}}{\ge} \sum_{T \in I^\dagger} \psi_T(\eta_T) - \frac{\eps}{2} \\
&\overset{\text{\cref{eq:psi-approx}}}{\ge} \sum_{T \in I^\dagger} \Phi^{(r)}(T,\eta_T) - \eps 
\overset{\text{Lem \ref{lem:decomp}}}{=} \Phi^{(r)}(S^\dagger,\sigma_{S^\dagger}) -\eps,
\end{align*}
as claimed.
\end{proof}

\section{Hardness Result}
\label{sec:hardness}

We establish computational hardness of {$k$-\infmax} for low-temperature Ising models from the hardness of estimating the marginal probabilities of single vertices, which is a direct consequence of hardness of approximate counting \cite{Sly10,SS14,GSV16} and self-reducibility.

\begin{theorem}[\cite{Sly10,SS14,GSV16}]
\label{thm:marginal}
Suppose $\Delta \ge 3$ is an integer and $\delta > 0$ is a real. 
Assuming $\mathsf{RP}\neq \mathsf{NP}$, there is no $\mathsf{FPRAS}$ for the following problem: Given an Ising model on a graph $G=(V,E)$
from the family $\MM(\Delta,1+\delta)$ and a vertex $v \in V$, estimate $\Pr(X_v = +)$.
\end{theorem}

\begin{proof}
Given an $\mathsf{FPRAS}$ for estimating marginals, one can approximate the partition function efficiently. More specifically, suppose $V = \{v_1,\dots,v_n\}$ and we have
\[
\Pr(X_1=+1,\dots,X_n=+1) = \prod_{i=1}^n \Pr(X_i = +1 \mid X_1 = +1, \dots, X_{i-1}=+1).
\]
Each $\Pr(X_i = +1 \mid X_1 = +1, \dots, X_{i-1}=+1)$ corresponds to the marginal at $v_i$ in an Ising model on the subgraph induced by $\{v_i,\dots,v_n\}$ where the pinning on $\{v_1,\dots,v_{i-1}\}$ becomes external fields. Thus, we can approximate $\Pr(X_1=+1,\dots,X_n=+1)$ and hence the partition function via \cref{eq:Gibbs}.
We therefore deduce the theorem from the hardness results for computing the partition function in low-temperature Ising models \cite{Sly10,SS14,GSV16}.
\end{proof}

We now give the proof of \cref{thm:hardness}.

\begin{proof}[Proof of \cref{thm:hardness}]
We may assume without loss of generality that $\delta \le 1$.
Given a polynomial-time algorithm for $k$-{\infmax} 
for the family $\MM(\Delta,1+\delta)$ and $1$-bounded vertex weights, we show how to efficiently estimate $\Pr(X_v = +)$ for an Ising model on a graph $G=(V,E)$
from the family $\MM(\Delta,1+\delta)$ and a vertex $v \in V$. 

Define a graph $G'$ which is the disjoint union of $G$ and $k$ distinct isolated vertices $u_1,\dots,u_k$. 
Each $u_i$ has the same external field $h(u_i) = x$ which we can choose freely. 
Together with $\beta \in \R^E$ and $h \in \R^V$ this defines an Ising model on $G'$ which is still in the family $\MM(\Delta,1+\delta)$.
Let $a \in \R^V$ be a 1-bounded vertex weight vector defined by $a(v)=1$, $a(u_i)=1$ for $i=1,\dots,k$ and $a(u)=0$ for all other vertices.

Consider the $k$-{\infmax} problem for the Ising model on $G'$ and the weight vector $a$.
Let $U = \{u_1,\dots,u_k\}$ and $W = \{v,u_1,\dots,u_{k-1}\}$.
For a subset $S$ of vertices, let $\boldsymbol{+}_S \in \{\pm 1\}^S$ denote the partial assignment that assigns $+$ to all vertices in $S$.
We claim that
\begin{align}\label{eq:max-UW}
\max_{\substack{S\subseteq V \cup U,\, |S| \le k \\ \sigma_S \in \{\pm 1\}^S}} \left\{ \Phi(S,\sigma_S) \right\}
&= \max\left\{ \Phi(U,\boldsymbol{+}_U), \Phi(W,\boldsymbol{+}_W) \right\}. 
\end{align}
To see this, consider a feasible pair $(S,\sigma_S)$. 
If $S = U$, then
\begin{align}
\Phi(U,\sigma_U) 
= \sum_{i=1}^k (\sigma_{u_i} - \E[X_{u_i}]) 
\le \sum_{i=1}^k (1 - \tanh x) 
= k(1-\tanh x) = \Phi(U,\boldsymbol{+}_U). \label{eq:S=U}
\end{align}
If $S \neq U$, then without loss of generality suppose $S \cap U = \{u_1,\dots,u_j\}$ where $j \le k-1$ and we have
\begin{align}
\Phi(S,\sigma_S)
&= \E[X_v \mid X_{S \setminus U} = \sigma_{S \setminus U}] - \E[X_v] + \sum_{i=1}^j (\sigma_{u_i} - \E[X_{u_i}]) \nonumber\\
&\le 1 - \E[X_v] + (k-1)(1-\tanh x)
= \Phi(W,\boldsymbol{+}_W). \label{eq:Snot=U}
\end{align}
Therefore, \cref{eq:max-UW} follows from \cref{eq:S=U,eq:Snot=U}.

Suppose the provided algorithm returns $(\hat{S},\sigma_{\hat{S}})$ which satisfies 
\[
\Phi(\hat{S},\sigma_{\hat{S}}) 
\ge \max_{\substack{S\subseteq V \cup U,\, |S| \le k \\ \sigma_S \in \{\pm 1\}^S}} \left\{ \Phi(S,\sigma_S) \right\} - \eps
= \max\left\{ \Phi(U,\boldsymbol{+}_U), \Phi(W,\boldsymbol{+}_W) \right\} - \eps.
\] 
If $\hat{S} = U$, then we deduce from \cref{eq:S=U} that
\[
\Phi(U,\boldsymbol{+}_U)
\ge \Phi(\hat{S},\sigma_{\hat{S}})
\ge \Phi(W,\boldsymbol{+}_W) - \eps, 
\]
implying $\E[X_v] \ge \tanh x - \eps$.
If $\hat{S} \neq U$, then we deduce from \cref{eq:Snot=U} that
\[
\Phi(W,\boldsymbol{+}_W)
\ge \Phi(\hat{S},\sigma_{\hat{S}}) 
\ge \Phi(U,\boldsymbol{+}_U) - \eps,
\]
implying $\E[X_v] \le \tanh x + \eps$.
Thus, by picking $\tanh x$ and applying binary search we can estimate $\E[X_v]$ efficiently with additive error $\eps$ with high probability.
This transforms to an estimator for $\Pr(X_v = +1)$ with multiplicative error $\eps$ since $\Pr(X_v = +1)$ is lower bounded when $\delta \le 1$.
\end{proof}

\bibliographystyle{alpha}
\bibliography{ref.bib}

\end{document}